\documentclass[10pt,a4paper]{IEEEtran}
\usepackage[utf8]{inputenc}
\usepackage{amsmath}
\usepackage{amsthm}
\usepackage{amsfonts}
\usepackage{amssymb}
\usepackage{graphicx}
\newtheorem{theorem}{Theorem}
\newtheorem{definition}[theorem]{Definition}
\newtheorem{lemma}[theorem]{Lemma}

\usepackage{setspace}
\usepackage{xcolor}
\usepackage{booktabs}
\usepackage{subcaption}
\usepackage{algorithm,algorithmicx,algpseudocode}

\title{Essence of Factual Knowledge}
\author{%
	Ruoyu Wang, Daniel Sun, Guoqiang Li, Raymond Wong, Shiping Chen\\
}

\begin{document}
	\maketitle
	
	\section{Introduction}
		Knowledge bases are collections of domain-specific and commonsense facts.
		Recently, the sizes of KBs are rocketing due to automatic extraction for knowledge and facts.
		For example, the number of facts in WikiData is up to 974 million!
		According to our observation, current KBs, especially domain KBs, show strong relevance in relations according to some topics\cite{AMIE,KGist}.
		These patterns can be used to conclude and infer for part of facts in the KBs.
		Therefore, the original KBs can be minimzed by extracting patterns and essential facts.
		
		In this paper, we introduce a framework for extracting knowledge essence and reducing overall volume of KBs by mining semantic patterns in relations.
		Facts are formalized as first-order predicates and patterns are induced as Horn rules.
		
		Table~\ref{tab:example-family} and Rule~(\ref{rule:family-father}), (\ref{rule:family-mother}) show an example of such extraction.
		By extracting the rules from listed facts, both table~\ref{tab:example-sub-father} and \ref{tab:example-sub-mother} can be inferred from other tables and then be removed.
		
		\begin{table*}[t]
			\caption{Partial data on server configurations and status}
			\label{tab:example-family}
			\begin{subtable}[t]{0.2\linewidth}
				\centering
				\caption{parent/2}
				\label{tab:example-sub-parent}
				\begin{tabular}{cc}
					\toprule
					parent  & child\\
					\midrule
					james & harry\\
					lily & harry\\
					harry & sirius\\
					harry & albus\\
					ginny & sirius\\
					ginny & albus\\
					\bottomrule
				\end{tabular}
			\end{subtable}
			\begin{subtable}[t]{0.2\linewidth}
				\centering
				\caption{father/2}
				\label{tab:example-sub-father}
				\begin{tabular}{cc}
					\toprule
					father & child\\
					\midrule
					james & harry\\
					harry & sirius\\
					harry & albus\\
					\bottomrule
				\end{tabular}
			\end{subtable}
			\begin{subtable}[t]{0.2\linewidth}
				\centering
				\caption{mother/2}
				\label{tab:example-sub-mother}
				\begin{tabular}{cc}
					\toprule
					mother & child\\
					\midrule
					lily & harry\\
					ginny & sirius\\
					ginny & albus\\
					\bottomrule
				\end{tabular}
			\end{subtable}
			\begin{subtable}[t]{0.15\linewidth}
				\centering
				\caption{male/1}
				\label{tab:example-sub-male}
				\begin{tabular}{c}
					\toprule
					person\\
					\midrule
					james\\
					harry\\
					albus\\
					sirius\\
					\bottomrule
				\end{tabular}
			\end{subtable}
			\begin{subtable}[t]{0.15\linewidth}
				\centering
				\caption{female/1}
				\label{tab:example-sub-female}
				\begin{tabular}{c}
					\toprule
					person\\
					\midrule
					lily\\
					ginny\\
					\bottomrule
				\end{tabular}
			\end{subtable}
		\end{table*}
	
		\begin{table*}
			\begin{align}
				father(X, Y) & \gets parent(X, Y), male(X)\label{rule:family-father}\\
				mother(X, Y) & \gets parent(X, Y), female(X)\label{rule:family-mother}
			\end{align}
		\end{table*}
		
		The remaining is organized as follows: Section~\ref{sec:rule} analysed properties of rules as equivalence classes.
		Essence extraction problem is formally defined in Section~\ref{sec:problem}.
		And Section~\ref{sec:framework} introduces the basic framework for essence extraction.
		Finally Section~\ref{sec:conclusion} concludes the paper.
	
	\section{Properties of Horn Rules}\label{sec:rule}
		\subsection{Semantic Length and Fingerprint of a Rule}
			First-order Horn rules are adopted in our technique to describe semantic patterns in relations.
			They can further be decomposed into equivalence classes.
			Elements in each of the classes are arguments that are assigned to the same variables, and if some argument is assigned to constants, then the corresponding equivalence class only consists of the argument and the constant.
			For example, Rule~(\ref{rule:family-father}) is decomposed to the following equivalent classes (number in the brackets
denotes the argument index of certain relation, starting from 0):
			\begin{align*}
				X: &\{father[0], parent[0], male[0]\} & father[0]\\
				Y: &\{father[1], parent[1]\} & father[1]
			\end{align*}
			
			The length of a rule is defined by the follwoing equation:
			\begin{equation}
				|r| = \sum_{i} (|C_i| - 1)
			\end{equation}
			where $C_i$ is one of these equivalence classes.
			
			Fingerprints of rules are based on the equivalence classes with labels of arguments in the head that it applies to.
			For example, the last column of the above example shows the label of head arguments.
			
			\begin{lemma}\label{lemma:fingerprint}
				Two rules are semantically equivalent if and only if their fingerprints are identical.
			\end{lemma}
			\begin{proof}
				(Necessity)If two rules are semantically equivalent, they can be written in syntactically identical form. Thus equivalence classes of corresponding variables or constants are identical.
				
				(Sufficiency)Each equivalence class tells position of one variable. Therefore, equivalence of all classes ensures that the set of predicates in both rules are identical. The labels of head arguments further determine that the head predicates are the same. Thus, the two rules are identical.
			\end{proof}
		
		\subsection{Search Space for Rules}
			Let $\Omega$ be the search space for first-order Horn rules.
			Some elements in $\Omega$ make no sense and should be excluded.
			If some predicate in the body is identical to the head, then the predicate in the body is redundant.
			These rules are trivial rules.
			If some subset of the body does not share any variable with the remaining part (include the head), then the rule is either redundant nor unsatisfiable.
			The subset is called independent fragment.
			The new search space excluding these two types of rules is written as $\Omega_m$.
	
		\subsection{Extension on Rules}
			\begin{definition}[Limited Variable, Unlimited Variable, Generative Variable]
				A variable is unlimited in some Horn rule $r$ if there is only one argument in $r$ that is assigned to it.
				A variable is limited in $r$ if there exist at least two arguments in $r$ that are assigned to it.
				A variable is generative if there exist arguments in both the head and body of $r$ that are assigned to it.
			\end{definition}
		
			Searching for rules starts from most general forms, i.e. rules only with head predicate and arguments in the predicate are all unique unlimited variables.
			To construct new rules, new equivalence conditions are added to the equivalence classes.
			Syntactically, these operations fell in five extension operations, which is noted by $ext(r)$:
			
			\noindent \textbf{Case 1}: Assign an existing limited variable to some argument.
			
			\noindent \textbf{Case 2}: Add a new predicate with unlimited variables to the rule and then assign an existing limited variable to one of these arguments.
			
			\noindent \textbf{Case 3}: Assign a new limited variable to a pair of arguments.
			
			\noindent \textbf{Case 4}: Add a new predicate with unlimited variables to the rule and then assign a new limited variable to a pair of arguments.
			In this case, the two arguments are not both selected from the newly added predicate.
			
			\noindent \textbf{Case 5}: Assign a constant to some argument.
			
			According to the rule extension, $\forall r, r_e \in \Omega_m$, if $r_e \in ext(r)$, then $r_e$ is the extension of $r$, and $r$ is the origin of $r_e$ (denoted as $r \in ext^{-1}(r_e)$ since one may have multiple origins). Neighbours of a rule in $\Omega_m$ consist of all its extensions and origins. The above extension operations can be used to search on $\Omega_m$. Let $S = \{r|r$ has only a head predicate $p$ and all arguments of $p$ are unlimited variables$\}$, every element in $\Omega_m$ can be searched from some $r_0 \in S$. To prove this we define a property \textit{link} between predicates in a certain rule: If two predicates $p$ and $q$ in a rule $r$ share a limited variable $X$, then $p$ and $q$ are linked by $X$ in $r$, written as $p \diamond_X q$, or in short $p \diamond q$. Moreover, if there is a sequence of predicates $p \diamond w_0 \diamond \dots \diamond w_k \diamond q$, then there is a \textit{linked path} between $p$ and $q$, written as: $p \leftrightarrow^\diamond q$. With this property, we can prove the search completeness as follows:
			
			\begin{lemma}\label{lemma:linked}
				$\forall r \in \Omega_m$, every predicate in $r$ has a linked path with the head of $r$.
			\end{lemma}
			\begin{proof}
				Suppose a predicate $p$ in rule $r$ has no linked path with the head. Then $p$ is not itself the head. Let $P = \{q| p \leftrightarrow^\diamond q\}$, every predicate in $P$ has no linked path with the head. Then the fragment noted by $P$ does not share any variables with remaining predicates. Namely, $P$ denotes an independent fragment in rule $r$. According to the definition of $\Omega_m$, we have $r \not \in \Omega_m$, which contradicts with $ r \in \Omega_m$.
			\end{proof}
			
			\begin{lemma}\label{thm:searchable}
				(Search Completeness)Let $S = \{r|r$ has only a head predicate $p$ and all arguments of $p$ are unlimited variables$\}$, $\forall r \in \Omega_m, \exists r_0, r_1, \dots, r_n \in \Omega_m$, such that $r_0 \in S, r_1 \in ext(r_0), \dots, r \in ext(r_n)$.
			\end{lemma}
			\begin{proof}
				Suppose $p \diamond_X q$ in $r$. During the searching process of $r$, when $p$ is already in a intermediate status $r'$, an extension of $r'$ can be constructed by adding a new predicate $q$ and turning corresponding variables to $X$. Thus, predicate $q$ is introduced into $r'$. Therefore, if $w \leftrightarrow^\diamond q$ and $w$ is already in a intermediate status, then $q$ can be introduced into $r'$. According to Lemma~\ref{lemma:linked}, all predicates in $r$ has linked path with its head. Each predicate can be introduced into the rule iteratively starting from the head predicate where arguments are all different unlimited variables. Other limited variables and constants can be added to the rule by other extension operations to finally construct $r$.
			\end{proof}
		
			Rules with independent fragments will not be constructed starting from $S$, as the extension operations do not introduce new predicates without any shared variables with other predicates.
	
	\section{Problem Definition}\label{sec:problem}
		\begin{definition}[Essential Knowledge Extraction]
			Let $B$ be the original KB, which is a finite set of atoms.
			The extraction on $B$ is a triple $(H, N, C)$, where $H$ (for ``Hypothesis'') is the set of first-order Horn rules, $N$ (for ``Necessary'') is a subset of $B$, and $C$ (for ``Counter Examples'') is a subset of the complement of $B$ subject to CWA.
			$B, H, N, C$ satisfies ($\models$ is logical entailment):
			\begin{itemize}
				\item $N \land H \models (B \setminus N) \cup C$
				\item $\forall e \not \in B \cup C, N \land H \not \models e$
				\item $|N| + |C| + |H|$ is minimal
			\end{itemize}
			where $|N|$ is the number of predicates in $N$, and so be $|C|$.
			$|H|$ is defined as the sum of lengths of all rules in it.
		\end{definition}
		
		\begin{definition}[Minimum Vertex Cover Problem]
			Let $\mathcal{G}_{vc} = \langle V_{vc}, E_{vc}\rangle$ be an undirected graph. A minimum vertex cover $V_c$ of $\mathcal{G}_{vc}$ is a minimum subset of $V_{vc}$ such that $(u, v)\in E_{vc} \implies u \in V_c \lor v \in V_c$.
		\end{definition}
	
		Complexity of essence extraction can be proved by reducing minimum vertex cover problem to relational compression. Let $\mathcal{G}_{vc} = \langle V_{vc}, E_{vc} \rangle$ be the graph in the vertex cover problem.
		By the following settings we create a relational knowledge base aligning with $\mathcal{G}_{vc}$:
		Let $v$ be a unary predicate in $B$ for each $v \in V_{vc}$;
		let $edge$ be a unary predicate in $B$ for edges;
		add two constants $e_{ij}$ and $e'_{ij}$ to $C$ and six predicates $edge(e_{ij})$, $edge(e'_{ij})$, $v_i(e_{ij})$, $v_i(e'_{ij})$, $v_j(e_{ij})$, $v_j(e'_{ij})$ to $B$ for each $(v_i, v_j) \in E_{vc}$;
		add the following predicates to $B$: $edge(c_1), edge(c_2), \dots, edge(c_{2\cdot|E_{vc}|+1})$;
		and add the following constants to $C$: $d_1, d_2, \dots, d_{4\cdot|E_{vc}|+1}$.
		
		\begin{figure}[!t]
			\includegraphics[width=\linewidth]{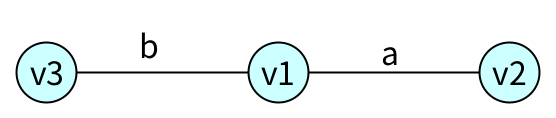}
			\caption{Vertex Cover Example}
			\label{fig:vertex-cover-example}
		\end{figure}
		
		For example, Figure~\ref{fig:vertex-cover-example} shows a graph with three vertices and two edges. The corresponding setting of relational compression is as follows:
		\begin{itemize}
			\item $C = \{a, a', b, b', c_1, \dots, c_5, d_1, \dots, d_9\}$
			\item $B = \{v_1(a)$, $v_1(a')$, $v_2(a)$, $v_2(a')$, $v_1(b)$, $v_1(b')$, $v_3(b)$, $v_3(b')$, $edge(a)$, $edge(a')$, $edge(b)$, $edge(b')$, $edge(c_1)$, $\dots$, $edge(c_5)\}$
		\end{itemize}
		
		By reducibility from minimum vertex cover problem to relational compression we can prove the latter is NP-hard. The details are as follows:
		
		\begin{lemma}\label{lemma:no-axiom-to-edge}
			$edge(X) \leftarrow true$ is not in $H$.
		\end{lemma}
		\begin{proof}
			Let $arg^+(p) = \{c \in C | p(c) \in B\}$, then $|arg^+(edge)| = 2n + 2n + 1 = 4n+1$, where $n = |E_{vc}|$. Thus, the number of predicates this rule entails is $4n+1$. Taking constants $d_1, \dots, d_{4n+1}$ into consideration, the number of counter examples this rule entails is also $4n+1$. The size reduced is $4n+1 - (4n+1) - 1 = -1$, no actual reduction. Therefore, it does not reduce the size of knowledge base. It is not in $H$.
		\end{proof}
		
		\begin{lemma}\label{lemma:edge-only-provable-by-vi}
			Predicates of $edge$ can only be entailed by the following rules: $edge(X) \leftarrow v_i(X)$.
		\end{lemma}
		\begin{proof}
			Let rule $r_i$ be: $edge(X) \leftarrow v_i(X)$, the length of which is 1. Then the number of predicates it entails is $2k$, where $k$ is the number of edges connected to vertex $v_i$. There are no counter examples entailed by this rule. Thus the size it reduces is $2k - |r_i| = 2k - 1$. If $k \ge 1$, this rule can be used to reduce the size of knowledge base.
			
			According to Lemma~\ref{lemma:no-axiom-to-edge}, $edge$ cannot be entailed by axioms, and since there is no other predicate in $B$, $edge$ can only be entailed by some $v_i$.
		\end{proof}
		
		\begin{lemma}\label{lemma:all-edges-provable}
			Let $S = \{edge(e)|edge(e) \in B\} \setminus \{edge(c)|\exists c_i = c\}$. All predicates in $S$ are provable after compression. That is, $S \subseteq R$, where $R$ is the set of all provable predicates.
		\end{lemma}
		\begin{proof}
			According to Lemma~\ref{lemma:edge-only-provable-by-vi}, proof of $edge(e) \in S$ relies only on predicates of $v_i$. No matter predicates of $v_i$ is provable or not, the rules of $edge(X) \leftarrow v_i(X)$ can always be applied to prove $edge(e) \in S$. Suppose $\exists edge(e) \in S$ such that $edge(e) \not \in R$. Then there is another predicate $edge(e') \in S$ and $edge(e') \not \in R$, where $e$ and $e'$ correspond to some edge in $E_{vc}$ and its duplicate, since these two predicates are both entailed by some rule $edge(X) \leftarrow v_i(X)$ if one of them is entailed by the rule. Then a new rule can be applied to entail these two predicates to further reduce the size of given result.
			However, according to definition of relational compression, output cannot be further reduced.
			Contradiction occurs.
		\end{proof}
		
		\begin{lemma}\label{lemma:vc-minimally-entails-edges-in-compression}\setstretch{1.1}
			Let $V_c$ be the solution of minimum vertex cover problem.
			Let $H_{vc}$ be a rule set and $H_{vc} = \{edge(X) \leftarrow v(X) | v \in V_c\}$.
			Let $\bar{H}_{vc}$ be a rule set and $\bar{H}_{vc} = \{edge(X) \leftarrow v(X) | v \not \in V_c\}$.
			Then $H_{vc} \subseteq H$ and $\bar{H}_{vc} \cap H = \varnothing$.
		\end{lemma}
		\begin{proof}
			According to Lemma~\ref{lemma:edge-only-provable-by-vi} and \ref{lemma:all-edges-provable}, all edges are provable and only provable by vertices, and this is equal to the setting that all edges are covered and only covered by vertices for minimum vertex cover problem. Thus $H_{vc}$ entails $S$ in a minimum cost.
			$H_{vc} \subseteq H$ and $\bar{H}_{vc} \cap H = \varnothing$.
		\end{proof}
		
		\begin{theorem}\label{np}
			Relational compression is NP-hard.
		\end{theorem}
		\begin{proof}
			%				According to the reduction setting,	solution of knowledge base compression problem can be polynomially converted to solution of minimum vertex cover problem.
			Let $V_c$ be the set of minimum vertex cover of $\mathcal{G}_{vc}$.
			According to the lemmas above, $V_c = \{v \in V_{vc} | \exists edge(X) \leftarrow v(X) \in H\}$.
			All the operations involved with reducibility are with polynomial cost.
			Thus minimum vertex cover problem can be polynomially reduced to relational compression. Relational compression is NP-hard.
		\end{proof}
	
	\section{Extraction Framework}\label{sec:framework}
	
		To tell whether a fact is provable by others, we employ a directed graph $\mathcal{G} = \langle V, E \rangle$ to encode dependency among facts with respect to inference.
		$V = B \cup \{\top\}$, where each vertex is either a fact in $B$ or an assertion of truth under no condition.
		$(b, h) \in E$ if $b$ is involved in the proof of $h$ by some rule.
		$(\top, h) \in E$ if $h$ can be inferred by some rule with empty body.
		The extraction for essence is given by Algorithm~\ref{alg:extraction}.
	
		\begin{algorithm}[!t]
			\caption{Essence Extraction}
			\label{alg:extraction}
			\begin{algorithmic}[1]
				\Require Knowledge Base $B$
				\Ensure Summarization on $B$: $(H, N, C)$
				\State $H \gets \varnothing$
				\State $C \gets \varnothing$
				\State $\mathcal{G} \gets \langle B \cup \{\top\}, \varnothing \rangle$
				\While {$r \gets findSingleRule(B)$}
					\State $H \gets H \cup \{r\}$
					\State $C \gets C \cup E^-_r$
					\State Update graph $\mathcal{G}$ with respect to $r$
				\EndWhile
				\State $cc \gets CoverCycle(\mathcal{G})$
				\State $N \gets \{h \in V \setminus \{\top\}| \forall b \in V, (b, h) \not \in E\} \cup cc$
				\State \Return $(H, N, C)$
			\end{algorithmic}
		\end{algorithm}
		
		If the dependency graph is a DAG, then essential predicates are represented by the vertices with zero in-degree.
		However, if cycles appear in $\mathcal{G}$, then at least one vertex in each cycle should be included in $N$.
		This assertion is proved bellow:
		
		\begin{lemma}\label{lemma:single-cycle}
			If some cycle in $\mathcal{G}$ is not overlapping with other cycles, then at least one vertex should be included in $N$.
		\end{lemma}
		\begin{proof}
			A vertex in the dependency graph is guaranteed provable if it is in $N$ or all of its in-neighbours are guaranteed provable.
			In the following proof, we assume that all other parts in $\mathcal{G}$ are guaranteed provable except the cycles.
			If none of vertices in a single cycle (not overlapping with other cycles) is included in $N$, then for each of these vertices, there is one in-neighbour not guaranteed provable.
			Thus, none of vertices in the cycle is guaranteed provable.
			At least one vertex should be selected in $N$.
		\end{proof}
		
		\begin{lemma}\label{lemma:overlapping-cycle}
			If some cycles in $\mathcal{G}$ are overlapping, then at least one vertex should be included in $N$.
		\end{lemma}
		\begin{proof}
			Suppose two cycles are overlapping in $\mathcal{G}$.
			If none of vertices in these cycles is in $N$, then none of them are guaranteed provable.
			If one of the vertices in the non-overlapping part is in $N$, then from this vertex to the one before intersection, all of these vertices are guaranteed provable.
			The other cycle is remained equivalent to circumstances of non-overlapping cycle and at least one of these vertices should be in $N$.
			If one of the vertices in the overlapping part is in $N$, then both cycles are guaranteed provable.
			In this case, still, at least one vertex is selected in each cycle.
			Cases are similar for more than two over lapping cycles.
		\end{proof}
		
		\begin{lemma}
			If there are cycles in the dependency graph, then at least one vertex should be included in $N$.
		\end{lemma}
		\begin{proof}
			It is clear by Lemma~\ref{lemma:single-cycle} and \ref{lemma:overlapping-cycle}.
		\end{proof}
		
		In the framework, two components may be implemented in different strategies according to specific domains: \textit{findSingleRule} and \textit{CoverCycle}.
		To implement \textit{findSingleRule}, pruning techniques are needed as the search space is large and useful candidates are sparse in the space.
		Given that semantic correlations may be strong in domain specific KBs, cycles are predicted to be large and frequent.
		Therefore, efficient coverage procedure is also required in the framework.
		
	\section{Restore the Original KB}
		As the dependency graph implies, if all the in-neighbours of some vertex in $\mathcal{G}$ is in the KB, the vertex can be inferred by applying some rule in $H$.
		Thus, in order to restore the original KB, we can iteratively apply each rule on current database until there is no more records inferred.
		Inference by a single rule can be done without full join in the relational data model.
		The algorithm is shown in Algorithm~\ref{alg:single-inference}.
			
		\begin{algorithm}[!t]
			\caption{Inference by a Single Pattern}
			\label{alg:single-inference}
			\begin{algorithmic}[1]
				\Require KB $B$
				\Require Rule $r$
				\Ensure Inferred target records $T$
				\State $C \gets $ equivalence classes of columns determined by $r$
				\State $B' \gets \varnothing$
				\For {body functor $f$ in $r$}
					\State $B_r \gets $ predicates in $f$ that complies to constant restrictions in $r$
					\State $B' \gets B' + B_r$ \label{line:preparing-data}
				\EndFor
				\For {equivalence class $c \in C$}
					\State Filter $B'$ by columns in $c$
					\State Update indices in $B'$
				\EndFor
				\State $T \gets \varnothing$
				\For {grouped row $w$ in $B'$}
					\State $w_t \gets $ empty relation for inferred records
					\For {target column $l$}
						\State $l_c \gets $ one condition column in $r$ that is equivalent to $l$
						\State $w_t.l \gets l_c$
					\EndFor
					\If {there are unassigned target columns in $w_t$}
						\State Assign each of these columns of all values in $B$
					\EndIf
					\State $T \gets T + w_t$
				\EndFor
				\State \Return $T$
			\end{algorithmic}
		\end{algorithm}
		
		The cost of single inference is proportional to the number of equivalence classes and to the size of relevant relations.
		The number of equivalence classes is proportional to the number of columns in $B'$.
		The cost for single inference is:
		\begin{equation*}
			|C| \cdot |B'| \propto |B'| \cdot \sum_{f \in B'} \phi(f)
		\end{equation*}
		where $\phi(f)$ is the arity of functor $f$.
		From the implication of the dependency graph, inference operations are the same as visiting vertices along paths in the graph.
		Thus, the maximum number of iterations is no larger than the maximum length of simple paths in $\mathcal{G}$.
		The overall cost of decompression is:
		\begin{equation*}
			d_{max} \cdot |P| \cdot |B| \cdot \sum_{f \in B} \phi(f)
		\end{equation*}
		where $d_{max}$ the maximum length of simple path in $\mathcal{G}$.
		
		\begin{lemma}
			When $d_{max}$ has reached to its maximum, the worst case cost of restoring is $O(|B|^3)$.
		\end{lemma}
	
		\begin{proof}
			According to the definition, $d_{max} \le |B| - 1$.
			When $d_{max} = |B| - 1$, all vertices in $\mathcal{G}$ form one single simple path.
			In this case, there can only be one rule in $H$ and only one relation in $B$.
			And the maximum number of arguments is also $|B| - 1$, otherwise the rule cannot summarize $B$.
			Therefore, the worst case cost is:
			$$(|B| - 1) \cdot 1 \cdot |B| \cdot (|B| - 1) = O(|B|^3)$$
			Other cases are the same.
		\end{proof}
	
	\section{Conclusion}\label{sec:conclusion}
		In this paper, we introduced a framework for extracting essence from factual knowledge.
		Theoretical proofs are also given for key properties of the framework.
		To put it into practice, more concrete work is required to design and analyze in \textit{findSingleRule} and \textit{CoverCycle}.
	
	\bibliographystyle{IEEEtran}
	\bibliography{Proof}

% Generated by IEEEtran.bst, version: 1.14 (2015/08/26)
\begin{thebibliography}{1}
\providecommand{\url}[1]{#1}
\csname url@samestyle\endcsname
\providecommand{\newblock}{\relax}
\providecommand{\bibinfo}[2]{#2}
\providecommand{\BIBentrySTDinterwordspacing}{\spaceskip=0pt\relax}
\providecommand{\BIBentryALTinterwordstretchfactor}{4}
\providecommand{\BIBentryALTinterwordspacing}{\spaceskip=\fontdimen2\font plus
\BIBentryALTinterwordstretchfactor\fontdimen3\font minus
  \fontdimen4\font\relax}
\providecommand{\BIBforeignlanguage}[2]{{%
\expandafter\ifx\csname l@#1\endcsname\relax
\typeout{** WARNING: IEEEtran.bst: No hyphenation pattern has been}%
\typeout{** loaded for the language `#1'. Using the pattern for}%
\typeout{** the default language instead.}%
\else
\language=\csname l@#1\endcsname
\fi
#2}}
\providecommand{\BIBdecl}{\relax}
\BIBdecl

\bibitem{AMIE}
L.~A. Gal{\'a}rraga, C.~Teflioudi, K.~Hose, and F.~Suchanek, ``Amie:
  association rule mining under incomplete evidence in ontological knowledge
  bases,'' in \emph{Proceedings of the 22nd international conference on World
  Wide Web}, 2013, pp. 413--422.

\bibitem{KGist}
C.~Belth, X.~Zheng, J.~Vreeken, and D.~Koutra, ``What is normal, what is
  strange, and what is missing in an knowledge graph,'' in \emph{The Web
  Conference}, 2020.

\end{thebibliography}
\end{document}